\documentclass[conference]{IEEEtran}

\usepackage{cite,enumitem}
\usepackage[cmex10]{amsmath} 
\usepackage{amssymb,amsthm,verbatim}
\usepackage{tcolorbox} 
\usepackage{graphicx,mathtools}
\usepackage{tikz,soul}
\usepackage[english]{babel}
\usepackage{csquotes}
\usepackage{bbm,framed}
\usepackage{xcolor}
\usepackage{transparent}
\usepackage[utf8]{inputenc} 
\usepackage[T1]{fontenc}
\usepackage{url}
\usepackage{ifthen}

\definecolor{lightgray}{RGB}{224,224,224}

\newtheorem{theorem}{Theorem}
\newtheorem{exmp}{Example}
\newtheorem{definition}{Definition}
\newtheorem{remark}{Remark}
\newtheorem{conjecture}{Conjecture}
\newtheorem{prop}{Proposition}

\newcommand{\euler}{\mathrm{e}}

\begin{document}
\title{On the Redundancy of Function-Correcting Codes over Finite Fields}

\author{%
  \IEEEauthorblockN{Hoang Ly and Emina Soljanin}
  \IEEEauthorblockA{
  Department of Electrical \& Computer Engineering, Rutgers University}
                    E-mail: \{mh.ly,\ emina.soljanin\}@rutgers.edu
}
\maketitle

\begin{abstract}
Function-correcting codes (FCC) protect specific function evaluations of a message against errors. This condition imposes a less stringent distance requirement than classical error-correcting codes (ECC), allowing for reduced redundancy. FCC were introduced by Lenz~\emph{et al.}~(2021), who also established a lower bound on the optimal redundancy for FCC over the binary field. Here, we derive an upper bound within a logarithmic factor of this lower bound. We show that the same lower bound holds for any finite field. Moreover, we show that this bound is matched for sufficiently large fields by demonstrating that it also serves as an upper bound. Furthermore, we construct an encoding scheme that achieves this optimal redundancy. Finally, motivated by these two extreme regimes, we conjecture that our bound serves as a valid upper bound across all finite fields.

\end{abstract}


\section{Introduction}
Function-Correcting Codes (FCC), introduced by Lenz et al. in 2021~\cite{FCC:conf/isit/LenzBWY21,FCCs:journals/tit/LenzBWY23}, are a novel class of codes that allow the receiver to reliably recover a specific function (or feature) of a message without reconstructing the entire message. This paradigm should enable substantial redundancy reduction when only a particular message function (or feature) is needed. Lenz et al.\ established an equivalence between FCC and irregular-distance codes—codes defined by non-uniform, function-dependent distance constraints between codewords. Leveraging this relationship, they derived bounds on the optimal redundancy of FCC based on the redundancy requirements of the corresponding irregular-distance codes. They also analyzed these bounds for some particular functions. Some bounds were later improved and generalized in~\cite{FCC:journals/arxiv/GeXZZ25,FCC:journals/arxiv/RajputRHH2025}.

A key advantage of FCC is their reduced overhead in scenarios where full data recovery is unnecessary and only a specific attribute of the message—namely, the evaluation of a given function—is of interest, as in distributed computation, data storage, or image classification tasks~\cite{FCCs:journals/tit/LenzBWY23}. For example, Premlal and Rajan extended the concept to linear functions~\cite{FCC:journals/arxiv/PremlalR}. Similarly, Xia et al.\ applied the FCC principles to symbol pair read channels -- a channel model in which overlapping symbol pairs are read in storage systems~\cite{FCC:journals/tit/QingfengHB24}. Yaakobi et al.\ recently extended the analysis to $b$-symbol read channels, extending FCC to multisymbol reads in modern storage devices~\cite{FCC:journals/arxiv/SinghSY25}. 
Authors in~\cite{message_stream:conf/innovations/GuptaZ25,message_stream:journals/corr/abs-2407-06446} investigated a class of functions called \textit{linear streaming}, where the receiver only needs to compute a single linear function of the message in a single pass and computationally limited space. They constructed an encoding scheme in which the codeword length grows nearly linearly with the message dimension, ensuring that the receiver can correctly evaluate the linear function with high probability under adversarial error rates of up to $1/4 - \epsilon$ (for any  $\epsilon>0$).

As with classical ECC, relatively little is known about the exact redundancy necessary for FCC to exist. In~\cite{FCCs:journals/tit/LenzBWY23}, it was shown that determining the optimal redundancy for a given function $f$ can be reformulated as the problem of finding the largest independent set in a specific, function-dependent graph, a problem known to be NP-complete. Consequently, no general explicit expression for optimal redundancy is known in the literature. Moreover, no explicit upper bounds for optimal redundancy are known for general functions over the binary field. For finite fields larger than $\mathbb{F}_2$, explicit bounds on the optimal redundancy have not yet been established.

In this paper, we establish redundancy bounds for FCC. We provide an upper bound for binary field, staying within a logarithmic factor of the known lower bound~\cite{FCCs:journals/tit/LenzBWY23}, and extend this lower bound to any \( q \)-ary finite field. For sufficiently large fields, we prove that this lower bound is matched by constructing an optimal encoding scheme. Motivated by these results, we conjecture that our upper bound holds across all finite fields.

The remainder of this paper is organized as follows. Section~\ref{sec:Problem_statement} formally defines the FCC model and introduces the necessary notation. In this section, we provide an illustrative example showing that the established lower bound on the redundancy of binary FCC is tight and achievable. We conclude by summarizing our main contributions. Section~\ref{main_results} presents our main results on redundancy bounds for FCC over binary and larger finite fields. Finally, Section~\ref{sec:Conclusion} concludes the paper.
\section{Problem Statement}\label{sec:Problem_statement}
\subsection{Notation}
We use the standard algebra and coding theory notations. \( \mathbb{F}_q \) denotes a finite field over some prime or prime power \( q \), with \( \mathbb{F}_2 \) denoting the binary field, and \( \mathbb{F}_q^n \) refering to an \( n \)-dimensional vector space over \( \mathbb{F}_q \). A \( q \)-ary, linear code \( \mathcal{C} \) of length \( n \), dimension \( k \), and minimum distance (hamming) \( d \) is denoted as \( [n, k, d]_q \). It is a \( k \)-dimensional subspace of the \( n \)-dimensional vector space \( \mathbb{F}_q^n \). The Hamming weight of a codeword \( \mathbf{x} \) in \( \mathcal{C} \), which counts the number of non-zero elements, is denoted by \( w(\mathbf{x}) \), and $d_H(\mathbf{x},\mathbf{y})$ denotes the Hamming distance between two sequences $\mathbf{x}$ and $\mathbf{y}$. Furthermore, \( \mathbf{0}_k \) and \( \mathbf{1}_k \) denote the all-zero and all-one row vectors of length \( k \), respectively. 
The binary unit column vector \( \mathbf{e}_i \) has a 1 at position \( i \) and 0s elsewhere. Finally, $\log$ denotes the base-2 logarithm function, and $\euler \approx 2.718$ is Euler's number.

\subsection{System Model}
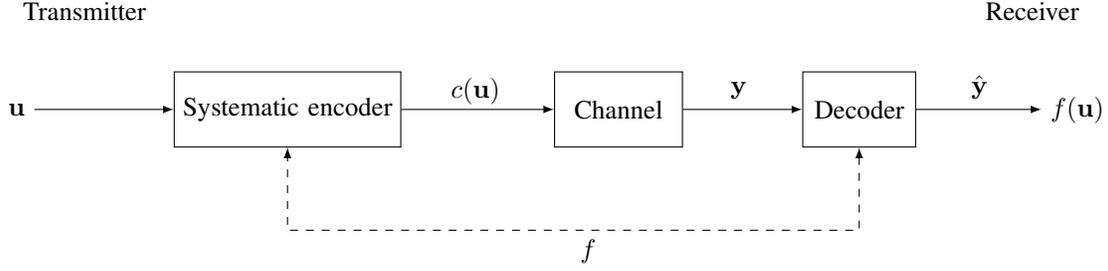
\begin{figure*}[]
    \centering
    \begin{tikzpicture}[>=latex]
    \node (u) at (-4.9,0) {\( \mathbf{u} \)};
    \node[draw, rectangle, minimum width=1.5cm, minimum height=1cm] (encoder) at (-1.3,0) {Systematic encoder};
    \node (p) at (1.2, 0.25) {\( c(\mathbf{u}) \)};
    \node[draw, rectangle, minimum width=1.7cm, minimum height=1cm] (channel) at (3.1,0) {Channel};
    \node (y) at (4.7, 0.25) {\( \mathbf{y} \)};
    \node[draw, rectangle, minimum width=1.5cm, minimum height=1cm] (decoder) at (6.3,0) {Decoder};
    \node (y) at (7.9, 0.3) {\( \hat{\mathbf{y}} \)};
    \node (D) at (9.2, 0) {\( f(\mathbf{u}) \)};
    
    \draw[->] (u) -- (encoder);

    \draw[->] (encoder) -- (channel);
    \draw[->] (channel) -- (decoder);
    \draw[->] (decoder) -- (D);

    \node at (-1.4, 1.3) {Transmitter};
    \node at (6.3, 1.3) {Receiver};

    \draw[dashed,<->] (decoder.south) -- ++(0,-1.1) 
    -- ++(-7.2,0) node[midway, below] {\( f \)} -- ++(-0.4,0) -- (encoder.south);

\end{tikzpicture}
    \caption{The transmitter has a message \( \mathbf{u} \), where the attribute \( f(\mathbf{u}) \) is of particular interest to the Receiver. To ensure the recoverability of this attribute, the transmitter encodes \( \mathbf{u} \) with a redundancy vector \( \mathbf{p}(\mathbf{u}) \) using a systematic encoder. 
Given a received vector \( \mathbf{y} \), which is an at-most-\( t \) erroneous version of the transmitted codeword  
\(c(\mathbf{u}) = (\mathbf{u}, \mathbf{p}(\mathbf{u})),\) and assuming knowledge of the function \( f \), the receiver can correctly compute \( f(\mathbf{u}) \).}
    \label{fig:system_model}
\end{figure*}

We adopt the system model proposed in~\cite{FCC:conf/isit/LenzBWY21,FCCs:journals/tit/LenzBWY23}, as illustrated in Figure~\ref{fig:system_model}. The transmitter has a message \(\mathbf{u} \in \mathbb{F}_q^k\), and the receiver wishes to evaluate a function $f$ on the message.
\[
f : \mathbb{F}_q^k \to \mathrm{Im}(f), ~ \text{where} ~
\mathrm{Im}(f) \triangleq \{ f(\mathbf{u}) : \mathbf{u} \in \mathbb{F}_q^k \}
\]
and \(\bigl|\mathrm{Im}(f)\bigr|\) is the cardinality of the image set of the function. 
The transmitter employs a \emph{systematic} encoder \[ c : \mathbb{F}_q^k \to \mathbb{F}_q^{k+r}, \] that maps each message $\mathbf{u}$ to a codeword of the form
\[
c(\mathbf{u}) \;=\; \bigl(\mathbf{u},\,\mathbf{p}(\mathbf{u})\bigr),
\]
where \(\mathbf{p}(\mathbf{u})\in \mathbb{F}_q^r\) is the added \emph{redundancy vector}, and \(r\) is called \emph{redundancy}. Thus, each codeword is a concatenation of the original message and its redundancy vector. 

The codeword \(c(\mathbf{u})\) is transmitted through a channel that may introduce up to \(t\) symbol errors. 
Therefore, the receiver observes a noisy copy of the transmitted codeword
\[
\mathbf{y} \;=\; c(\mathbf{u}) + \mathbf{e} \;\in\; \mathbb{F}_q^{k+r}, 
~ \text{with} ~ w(\mathbf{e}) \le t,
\]
where $\mathbf{e} \in \mathbb{F}_q^{k+r}$ is the noise vector and the addition is done in modulo $q$. We adopt the definition of FCC over any finite field from~\cite{FCC:conf/itw/PremlalR24}.
\begin{definition}[\hspace{-0.1mm}\cite{FCC:conf/itw/PremlalR24}]
A \emph{systematic} encoding \( c : \mathbb{F}_q^k \to \mathbb{F}_q^{k+r} \) is said to be an \((f,t)\)\emph{-FCC} for a function \( f : \mathbb{F}_q^k \to \mathrm{Im}(f) \) if, for all distinct pairs \(\mathbf{u}_i, \mathbf{u}_j \in \mathbb{F}_q^k\) with \( f(\mathbf{u}_i) \neq f(\mathbf{u}_j)\), the minimum distance condition
\begin{equation}\label{eq:codeword_distance}
    d_H\bigl(c(\mathbf{u}_i),\, c(\mathbf{u}_j)\bigr) \;\ge\; 2t + 1
\end{equation}
holds. Here, $t$ parameterizes the target correction capability.
\end{definition}

\begin{definition}[\hspace{-0.1mm}\cite{FCC:conf/itw/PremlalR24}]
The \emph{optimal redundancy} \( r_f(k,t) \) is the smallest \(r\) such that there exists an \((f,t)\)-FCC 
\[
c : \mathbb{F}_q^k \,\to\, \mathbb{F}_q^{k+r}.
\]
\end{definition}

\begin{remark}\label{rm:FCC_to_ECC}
The following properties, observed in~\cite{FCCs:journals/tit/LenzBWY23}, illustrate the relationship between FCC and classical ECC:
\begin{itemize}
    \item If \(f\) is \textbf{bijective}, then \(|\mathrm{Im}(f)| = |\mathcal{C}| = q^k\). In this case, every pair of codewords must satisfy a minimum distance of \(2t+1\), and the \((f,t)\)-FCC reduces to a systematic \([n,\;k,\;2t+1]_q\) error-correcting code (ECC). In other words, any ECC over the same message space can be viewed as a special case of an FCC. Therefore,
\[
r_f(k,t) \;\le\; v(k,t),
\]
where \(v(k,t)\) denotes the smallest integer such that a systematic \([k + v(k,t),\;k,\;2t+1]_q\) ECC exists.

    \item For a \textbf{constant function} \(f\), that is, $|\mathrm{Im}(f)| = 1$, no additional redundancy is required. Indeed, \(c(\mathbf{u})=\mathbf{u}\) trivially meets the definition of an FCC, yielding \(r_f(k,t)=0\).
\end{itemize}
\end{remark}

By definition, the receiver can correctly determine \(f(\mathbf{u})\) from any received vector \(\mathbf{y}\), provided \(\mathbf{y}\) differs from a valid codeword \(c(\mathbf{u})\) by at most \(t\) symbols and that \( f \) and \( c \) are known to the receiver.

Unlike classical ECC, where every pair of codewords must satisfy a minimum distance requirement, FCC allow codewords associated with messages evaluating to the same function value to be arbitrarily close in the Hamming distance. This relaxation can reduce the redundancy required compared to classical error correction. FCC are therefore not equivalent to just applying an ECC over a codebook of size \(|\mathrm{Im}(f)|\). Moreover, it has been shown that the required redundancy of FCC do not depend on the size of the function's image set ($|\mathrm{Im}(f)|$), but rather on whether messages that are close in Hamming distance evaluate to different function values~\cite{FCCs:journals/tit/LenzBWY23}. The following example illustrates this idea and provides a practical connection.
\begin{exmp}[Multi-input $\mathsf{OR}$ Function]\label{ex:ORfunction}
Consider a simple sensing setup where \(\mathbf{u} = u_1u_2\dots u_k \in \mathbb{F}_2^k\) records the presence (1) or absence (0) of signals from \(k\) sensors. Define
\[
f(\mathbf{u}) = u_1 \vee u_2 \vee \dots \vee u_k \in \{0,1\},
\]
which serves as a binary indicator of whether any signal is active. This function behaves as a multi-input $\mathsf{OR}$ operand. Clearly,
\[
f(\mathbf{u}) = 0 \iff \mathbf{u} = \mathbf{0}_k,\quad
f(\mathbf{u}) = 1 \text{ otherwise}.
\]

To assess redundancy, compare \(\mathbf{u}_1 = \mathbf{0}_k\) and \(\mathbf{u}_2 = 00\dots01\), which yield different function values. Then
\begin{align*}
2t + 1 
&\le 
d_H(c(\mathbf{u}_1), c(\mathbf{u}_2)) 
= 1 + d_H(\mathbf{p}(\mathbf{u}_1), \mathbf{p}(\mathbf{u}_2)),
\end{align*}
so
\[
d_H(\mathbf{p}(\mathbf{u}_1), \mathbf{p}(\mathbf{u}_2)) \ge 2t,
\]
implying at least \(2t\) redundancy symbols are required. This bound is achievable by setting
\[
\mathbf{p}(\mathbf{0}_k) = \mathbf{0}_{2t},\quad 
\mathbf{p}(\mathbf{u}) = \mathbf{1}_{2t} \text{ for } \mathbf{u} \ne \mathbf{0}_k.
\]
Under this encoding, \(c(\mathbf{u}) = \mathbf{0}_{k+2t}\) if and only if \(\mathbf{u} = \mathbf{0}_k\), and for all other \(\mathbf{u}\), the codeword weight satisfies \(w(c(\mathbf{u})) \ge 2t + 1\), ensuring the distance condition is met. Thus, \(r_f(k, t) = 2t\).

This example shows that the lower bound \(2t\) is tight and independent of the code dimension \(k\).

Consider now the case where $|\mathrm{Im}(f)| = 2^k$, that is, \(f\) is bijective, so the FCC effectively behaves as a standard ECC. We will show that, for \(k>1\), one necessarily has \(r_f(k,t) = v(k,t) > 2t\). In fact, from the previous argument, at least \( 2t \) redundancy is required, implying \( n = r+k > 2t + 1 \) for \( k > 1 \). Applying the Hamming bound~\cite{Coding:books/MacWilliamsS77}:
\begin{align*}
2^r 
\;& =\;
\frac{2^n}{2^k}
\;\ge\;
\sum_{j=0}^t \binom{n}{j}(2-1)^j
\;>\;
\sum_{j=0}^t \binom{2t+1}{j}\\
\;& =\;
\frac{1}{2} \left(\sum_{j=0}^{2t+1} \binom{2t+1}{j}\right)
\;=\;
\frac{1}{2} \cdot 2^{2t+1}
\;=\;
2^{2t},
\end{align*}
which implies that \(r > 2t\), where the first inequality comes from the Hamming bound for binary codes ($q = 2$). 

Hence \(r>2t\), implying $v(k, t) > 2t$. In other words, in the case of classical ECC, we need at least $2t+1$ redundancy to maintain a minimum distance $2t+1$ in our code whenever $k \ge 2$. This shows that having all messages except one evaluate to the same function value helps reduce redundancy compared to the classical ECC case. 
\end{exmp}

As demonstrated in the previous example, even when all messages except one give the same function value, the redundancy required is still at least $2t$. This observation was generalized in the following proposition.

\begin{prop}[\hspace{-0.1mm}\cite{FCCs:journals/tit/LenzBWY23}]\label{pro:lower_bound_binary}
Let \(f: \mathbb{F}_2^k \to \mathrm{Im}(f)\) be any function with \(\lvert\mathrm{Im}(f)\rvert \ge 2\). Then the optimal redundancy \(r_f(k,t)\) of an \((f,t)\)-FCC over binary field $\mathbb{F}_2$ satisfies
\[
r_{f}(k,t) \;\ge\; 2t.
\]
\end{prop}
\subsection{Summary of Results}
A central problem in the study of FCC is how to construct function-correcting schemes with minimal redundancy. In this paper, we investigated the redundancy requirements of FCC over finite fields, focusing on their ability to protect specific function evaluations of messages from errors. Our main contributions are as follows:
\begin{itemize}
    \item \textbf{Lower Bound on Redundancy:} We established that the optimal redundancy \( r_f(k, t) \) of an \((f, t)\)-FCC is at least \( 2t \) for any finite field \( \mathbb{F}_q \), provided the function \( f \) maps to at least two distinct values.
    \item \textbf{Upper Bound for Binary Field:} For the binary field \( \mathbb{F}_2 \), we derived an upper bound on \( r_f(k, t) \), showing that it grows logarithmically with the code dimension \( k \). 
    \item \textbf{Achievability for Large Fields:} We demonstrated that for sufficiently large fields (\( q \geq k + 2t \)), the lower bound \( r_f(k, t) = 2t \) is achievable. This result was achieved by constructing systematic Maximum Distance Separable (MDS) codes with minimum distance \( 2t + 1 \).
    \item \textbf{Conjecture for Moderate Field Sizes:} Motivated by the results for binary and large fields, we conjectured that the upper bound derived for binary field remains valid for all finite fields, even when \( q < k + 2t \).
\end{itemize}

These results provide a comprehensive understanding of the FCC's redundancy requirements for any field, highlighting the trade-offs among field size, code dimension, and redundancy.

\section{Main Results}\label{main_results}
In Subsection~A, we present an upper bound on the optimal redundancy of FCC over the binary field. In Subsection~B, we establish a lower bound on the optimal redundancy and show, by explicit construction, that this bound is achieved when the field size is sufficiently large. Based on these results, we conjecture in Subsection~C that the same upper bound holds for the optimal redundancy of FCC over all finite fields.

\subsection{Bounds on FCC over Binary Field}
\begin{theorem}\label{theo:bound_binary}
Let \(f: \mathbb{F}_2^k \to \mathrm{Im}(f)\) be a function with \(\lvert \mathrm{Im}(f) \rvert \ge 2\), and suppose \(k \ge 2\). Then the optimal redundancy \(r_f(k,t)\) of an \((f,t)\)-FCC over $\mathbb{F}_2$ satisfies
\begin{equation}\label{eq:theorem_bounds}
    2t \;\;\le\; r_f(k,t) \;<\; \frac{t \,\log \bigl(2k\bigr)}{1 \;-\; \frac{t}{k} \,\log \euler}\,.
\end{equation}
\end{theorem}

\begin{proof}
\textbf{Lower Bound.}  
The lower bound \(r_f(k,t) \ge 2t\) is exactly Proposition~\ref{pro:lower_bound_binary}.  A concise proof also appears in~\cite{FCCs:journals/tit/LenzBWY23}. 

\medskip
\textbf{Upper Bound.}  
We establish the strict inequality
\begin{equation}\label{eq:to_show_upper}
    r_f(k,t) \;<\; \frac{t \,\log\bigl(2k\bigr)}{1 \;-\; \frac{t}{k}\,\log \euler}.
\end{equation}
We rely on the fact in Remark~\ref{rm:FCC_to_ECC} that any \([n,k,2t+1]\) \emph{systematic} binary code can serve as an \((f,t)\)-FCC.

From classical coding theory (see, e.g., \cite[Ch.\,9]{Coding:books/MacWilliamsS77}), \emph{binary, systematic BCH codes} of length \(n\) and minimum Hamming distance \(2t+1\) exist with redundancy
\[
r \;\le\; \Bigl\lfloor t\, \log\bigl(n+1\bigr) \Bigr\rfloor \le t\, \log\bigl(n+1\bigr).
\]
Since \(n+1 = k+r + 1\), we obtain
\[
\log\bigl(n+1\bigr)
\;=\;
\log\Bigl(k + \bigl(r+1\bigr)\Bigr)
\;=\;
\log(k)\;+\;\log\!\Bigl(1+\tfrac{r+1}{k}\Bigr).
\]
Using the inequality \(\log(1+x)\le x \log \euler\) for \(x > 0\), it follows that
\begin{equation}\label{eq:bch_log_bound}
\log\bigl(n+1\bigr)
\;\le\;
\log(k)
\;+\;
\tfrac{r+1}{k}\,\log \euler.
\end{equation}
Rearrange the terms, hence
\[
r \;\Bigl(1 \;-\; \frac{t\,\log \euler}{k}\Bigr)
\;\le\;
t\,\log(k)
\;+\;
\frac{t\,\log \euler}{k},
\]
and so
\[
r 
\;<\;
\frac{t\,\log\bigl(2k\bigr)}{1 \;-\; \tfrac{t\,\log \euler}{k}},
\]
where we used the fact \(\frac{\log\euler}{k} < 1 \) for \(k\ge 2\). Hence there exists a systematic \([k+r,k,2t+1]\) binary code with \(r\) bounded exactly as in \eqref{eq:to_show_upper}.

\medskip
\emph{Concluding the Proof.}  
By Remark~\ref{rm:FCC_to_ECC}, any systematic linear code with minimum distance \(2t+1\) suffices to correct \(t\) symbol errors for \emph{all} pairs of distinct codewords.  Therefore, it also suffices to distinguish any two messages \(\mathbf{u}_i,\mathbf{u}_j\) with \(f(\mathbf{u}_i)\neq f(\mathbf{u}_j)\). Consequently, the redundancy of an \((f,t)\)-FCC, \(r_f(k,t)\), is at most the \(r\) of this code.  Hence
\[
r_f(k,t) 
\;\le\;
r 
\;<\; 
\frac{t\,\log\bigl(2k\bigr)}{1 \;-\; \tfrac{t}{k}\,\log \euler},
\]
completing the proof of \eqref{eq:theorem_bounds}.
\end{proof}
We observe that for fixed correction capability $t$ and increasing message dimension $k$, the upper bound remains within a logarithmic factor of the lower bound $2t$.
\subsection{Bounds and Achievability over Finite Fields}
\begin{theorem}\label{thm:bound_finite_field}
  Let \(f : \mathbb{F}_q^k \to \mathrm{Im}(f)\) be a function with \(\lvert \mathrm{Im}(f)\rvert \ge 2\). 
  Then the optimal redundancy \(r_f(k,t)\) of an \((f,t)\)-FCC over any field $\mathbb{F}_q$ satisfies
  \[
    r_f(k, t) \;\geq\; 2t.
  \]
Moreover, equality happens when \(q \ge k + 2t\).
\end{theorem}

\begin{proof}
\noindent\textbf{Lower Bound.} 
A version of this bound was proved for binary field in~\cite{FCCs:journals/tit/LenzBWY23}. 
We extend the result to any finite field \(\mathbb{F}_q\).

\smallskip
\emph{Step 1: Existence of a pair of messages differing in exactly one coordinate that map to two values of \(f\).}
Suppose, for contradiction, that every pair of messages whose coordinates differ in exactly one position must yield the same function value under \(f\). For each integer \(i\) with \(0 \le i \le k\), define the sets
\[
  A_i \;\triangleq\; \bigl\{\mathbf{u} \in \mathbb{F}_q^k : \text{the Hamming weight of \(\mathbf{u}\) is } i \bigr\}.
\]
Here, the Hamming weight of \(\mathbf{u}\) is the number of nonzero coordinates in \(\mathbf{u}\). By construction, the sets \(A_i\) are pairwise disjoint and their union is all of \(\mathbb{F}_q^k\).

Let \(f_0 \triangleq f(\mathbf{0}_k)\) be the value of \(f\) at the all-zero vector. By assumption, any vector in \(A_1\) differs from \(\mathbf{0}_k\) in exactly one coordinate and therefore must map to the same value \(f_0\). Next, every vector in \(A_2\) differs in exactly one coordinate from some vector in \(A_1\), implying all vectors in \(A_2\) also map to \(f_0\). Proceeding inductively, for each \(i\in\{1,\dots,k\}\), any \(\mathbf{u}\in A_i\) is at Hamming distance 1 from some \(\mathbf{v}\in A_{i-1}\). By the same assumption, \(f(\mathbf{u}) = f(\mathbf{v}) = f_0\). Consequently, all \(\mathbf{u}\in \mathbb{F}_q^k\) satisfy \(f(\mathbf{u})=f_0\). This implies that $f$ is a constant function and contradicts \(\lvert \mathrm{Im}(f)\rvert \ge 2\). Hence, there must exist \(\mathbf{u}_1,\mathbf{u}_2\in \mathbb{F}_q^k\) that differ in exactly one coordinate with \(f(\mathbf{u}_1)\neq f(\mathbf{u}_2)\).

\smallskip
\emph{Step 2: Distance requirement and redundancy lower bound.}
Since \(f(\mathbf{u}_1)\neq f(\mathbf{u}_2)\), any valid \((f,t)\)-FCC encoding \(c : \mathbb{F}_q^k \to \mathbb{F}_q^{k+r}\) must assign codewords \(c(\mathbf{u}_1), c(\mathbf{u}_2)\) that differ in at least \(2t+1\) coordinates:
\[
  d_H\bigl(c(\mathbf{u}_1),\, c(\mathbf{u}_2)\bigr) 
  \;\ge\;
  2t + 1.
\]
Since \(\mathbf{u}_1\) and \(\mathbf{u}_2\) differ in exactly one coordinate, \(d_H(\mathbf{u}_1,\mathbf{u}_2)=1\). Thus,
\begin{align*}
  2t + 1 
  \;&\le\;
   d_H\bigl(c(\mathbf{u}_1),\, c(\mathbf{u}_2)\bigr)\\
  \;& =\;
  d_H\bigl(\mathbf{u}_1, \mathbf{u}_2\bigr)
  \;+\;
  d_H\bigl(\mathbf{p}(\mathbf{u}_1),\, \mathbf{p}(\mathbf{u}_2)\bigr) \\
  & = 1 \;+\; d_H\bigl(\mathbf{p}(\mathbf{u}_1),\,\mathbf{p}(\mathbf{u}_2)\bigr),
\end{align*}
implying
\[
  d_H\bigl(\mathbf{p}(\mathbf{u}_1),\,\mathbf{p}(\mathbf{u}_2)\bigr) 
  \;\ge\; 
  2t.
\]
At least \(2t\) redundant symbols are therefore required for this single pair of messages to be separated by Hamming distance \(2t+1\). Consequently,
\[
  r_f(k,t) \;\ge\; 2t.
\]

\smallskip
\noindent\textbf{Achievability when \(q \ge k + 2t\).}  
When \(q \ge k + 2t\), it is possible to construct a systematic \(\bigl[n = k + 2t,\, k,\, 2t + 1\bigr]_q\) MDS code whose generator matrix is of the systematic form (see, e.g., \cite{DSS:journals/ftcit/RamkumarBSVKK22,MDS:journals/ComLet/LacanF04} for concrete constructions):
\[
\mathbf{G}_{k \times n} = [\,\mathbf{I}_{k} \,|\, \mathbf{P}\,],
\]
in which the first $k$ columns constitute an Identity matrix of size $k$, and the last $(n-k)$ columns are parity checks. Each message \(\mathbf{u} \in \mathbb{F}_q^k\) can be systematically encoded using $\mathbf{G}$ as
\[
  c(\mathbf{u}) \;=\; \mathbf{u}\cdot\mathbf{G},
\]
yielding a linear code whose minimum distance is \(2t+1\) and redundancy is \(r=n-k = 2t\). Because the constructed code ensures a minimum distance of at least \(2t+1\) between \emph{every pair} of distinct codewords, it is also sufficient for distinguishing every pair of messages that map to different function values under \(f\). This shows the existence of an FCC with redundancy $2t$, and therefore:
\[
  r_f(k,t) 
  \;\;\le\;\; 2t.
\]
Combining this with the established lower bound shows that
\[
  r_f(k,t) \;=\; 2t
~
\text{whenever}
~
  q \ge k + 2t.
\]
This completes the proof.
\end{proof}
Theorem~\ref{thm:bound_finite_field} implies that over sufficiently large fields ($q \ge k + 2t$), function-correcting codes (FCC) offer no redundancy advantage over classical error-correcting codes (ECC). In this regime, encoding with systematic MDS codes suffices to protect arbitrary function evaluations against errors, while still achieving the minimum possible redundancy. 
\begin{remark}
Theorem~\ref{thm:bound_finite_field} demonstrates that, over sufficiently large fields (i.e., \(q \ge k + 2t\)), the optimal redundancy \(r_f(k,t)\) equals \(2t\) and is independent of the code dimension \(k\). In contrast, for smaller alphabets (e.g., the binary field), the achievable redundancy may grow with \(k\), e.g., as fast as $\log(k)$; see, for example, Theorem~\ref{theo:bound_binary} and Example~\ref{ex:ORfunction}. Hence, there is a fundamental trade-off: larger fields allow redundancy that depends only on the correction capability \(t\), whereas smaller fields can force higher redundancy depending on \(k\). This interplay between field size and redundancy is a key consideration for practical code design.
\end{remark}

\subsection{A Conjecture on the Redundancy of FCC over Finite fields}
We note that whenever \(k \ge 2\),
\[
2t \le t\log(2k) < \frac{t \log(2k)}{1 - \tfrac{t}{k}\log \euler}.
\]
Therefore, the upper bound in~\eqref{eq:theorem_bounds} for the redundancy of FCC over the binary field provides a looser yet valid upper bound on the optimal redundancy of FCC over sufficiently large fields, specifically when \(q \ge k + 2t\). Examining the two extreme regimes of the field size—very large \(q\), where \(r_f(k,t) = 2t\), and very small \(q\), such as the binary field, where the redundancy may grow with \(k\)—motivates the following conjecture: we posit that for all “moderate” field sizes, an upper bound analogous to that in the binary case still applies.
\noindent

\begin{conjecture}
Let \(f : \mathbb{F}_q^k \to \mathrm{Im}(f)\) be a function with \(\lvert \mathrm{Im}(f)\rvert \ge 2\) and \(r_f(k,t)\) denote the optimal redundancy of an \((f,t)\)-FCC over a finite field \(\mathbb{F}_q\). Then, for all \(q < k + 2t\),
\[
  r_f(k,t) \;<\; \frac{t \,\log\bigl(2k\bigr)}{1 \;-\;\tfrac{t}{k}\,\log \euler}.
\] If that holds, then for all finite fields $\mathbb{F}_q^k$,
\[
  r_f(k,t) \;<\; \frac{t \,\log\bigl(2k\bigr)}{1 \;-\;\tfrac{t}{k}\,\log \euler}.
\]
\end{conjecture}

\noindent
In other words, we conjecture that for all finite fields the redundancy is still bounded above by the same quantity as in the binary-field setting. Verifying or refuting this behavior for all finite fields remains an open problem of particular interest.


\section{Conclusion}\label{sec:Conclusion}
This work explored the redundancy of FCC over finite fields, a class of codes designed to protect specific function evaluations rather than entire messages. We established tight bounds on the redundancy for the binary field and large finite fields, showing that the redundancy depends on the error correction capability \( t \) and, in some cases, the code dimension \( k \). The redundancy is independent of \( k \) for sufficiently large fields, achieving the optimal value \( 2t \).

Our findings highlight the efficiency of FCC in reducing redundancy compared to classical error-correcting codes, particularly in scenarios where only specific function evaluations need protection. The question of moderate field sizes remains an open and offers a promising direction for future research. 
\section*{Acknowledgment}
This work was supported in part by NSF CCF-2122400. We thank anonymous reviewers for their constructive feedback.

\bibliography{ISTC}

\begin{thebibliography}{10}
\providecommand{\url}[1]{#1}
\csname url@samestyle\endcsname
\providecommand{\newblock}{\relax}
\providecommand{\bibinfo}[2]{#2}
\providecommand{\BIBentrySTDinterwordspacing}{\spaceskip=0pt\relax}
\providecommand{\BIBentryALTinterwordstretchfactor}{4}
\providecommand{\BIBentryALTinterwordspacing}{\spaceskip=\fontdimen2\font plus
\BIBentryALTinterwordstretchfactor\fontdimen3\font minus
  \fontdimen4\font\relax}
\providecommand{\BIBforeignlanguage}[2]{{%
\expandafter\ifx\csname l@#1\endcsname\relax
\typeout{** WARNING: IEEEtran.bst: No hyphenation pattern has been}%
\typeout{** loaded for the language `#1'. Using the pattern for}%
\typeout{** the default language instead.}%
\else
\language=\csname l@#1\endcsname
\fi
#2}}
\providecommand{\BIBdecl}{\relax}
\BIBdecl

\bibitem{FCC:conf/isit/LenzBWY21}
\BIBentryALTinterwordspacing
A.~Lenz, R.~Bitar, A.~Wachter{-}Zeh, and E.~Yaakobi, ``Function-correcting
  codes,'' in \emph{{IEEE} International Symposium on Information Theory,
  {ISIT} 2021, Melbourne, Australia, July 12-20, 2021}.\hskip 1em plus 0.5em
  minus 0.4em\relax {IEEE}, 2021, pp. 1290--1295. [Online]. Available:
  \url{https://doi.org/10.1109/ISIT45174.2021.9517976}
\BIBentrySTDinterwordspacing

\bibitem{FCCs:journals/tit/LenzBWY23}
\BIBentryALTinterwordspacing
------, ``Function-correcting codes,'' \emph{{IEEE} Trans. Inf. Theory},
  vol.~69, no.~9, pp. 5604--5618, 2023. [Online]. Available:
  \url{https://doi.org/10.1109/TIT.2023.3279768}
\BIBentrySTDinterwordspacing

\bibitem{FCC:journals/arxiv/GeXZZ25}
\BIBentryALTinterwordspacing
G.~Ge, Z.~Xu, X.~Zhang, and Y.~Zhang, ``Optimal redundancy of
  function-correcting codes,'' \emph{CoRR}, vol. abs/2502.16983, 2025.
  [Online]. Available: \url{https://doi.org/10.48550/arXiv.2502.16983}
\BIBentrySTDinterwordspacing

\bibitem{FCC:journals/arxiv/RajputRHH2025}
\BIBentryALTinterwordspacing
C.~Rajput, B.~S. Rajan, R.~Freij-Hollanti, and C.~Hollanti,
  ``Function-correcting codes for $\rho$-locally $\lambda$-functions,'' 2025.
  [Online]. Available: \url{https://arxiv.org/abs/2504.07804}
\BIBentrySTDinterwordspacing

\bibitem{FCC:journals/arxiv/PremlalR}
\BIBentryALTinterwordspacing
R.~Premlal and B.~S. Rajan, ``Linear-function correcting codes,'' \emph{CoRR},
  vol. abs/2404.15135, 2024. [Online]. Available:
  \url{https://doi.org/10.48550/arXiv.2404.15135}
\BIBentrySTDinterwordspacing

\bibitem{FCC:journals/tit/QingfengHB24}
Q.~Xia, H.~Liu, and B.~Chen, ``Function-correcting codes for symbol-pair read
  channels,'' \emph{IEEE Transactions on Information Theory}, vol.~70, no.~11,
  pp. 7807--7819, 2024.

\bibitem{FCC:journals/arxiv/SinghSY25}
\BIBentryALTinterwordspacing
S.~Anamika, K.~S. Abhay, and Y.~Eitan, ``Function-correcting codes for b-symbol
  read channels,'' \emph{CoRR}, vol. abs/2503.12894, 2025. [Online]. Available:
  \url{https://arxiv.org/abs/2503.12894}
\BIBentrySTDinterwordspacing

\bibitem{message_stream:conf/innovations/GuptaZ25}
M.~Gupta and R.~Y. Zhang, ``Error correction for message streams,'' in
  \emph{16th Innovations in Theoretical Computer Science Conference, {ITCS}
  2025, January 7-10, 2025, Columbia University, New York, NY, {USA}}, ser.
  LIPIcs, R.~Meka, Ed., vol. 325.\hskip 1em plus 0.5em minus 0.4em\relax
  Schloss Dagstuhl - Leibniz-Zentrum f{\"{u}}r Informatik, 2025, pp.
  59:1--59:18.

\bibitem{message_stream:journals/corr/abs-2407-06446}
\BIBentryALTinterwordspacing
M.~Gupta, V.~Guruswami, and M.~Singhal, ``Tight bounds for stream decodable
  error-correcting codes,'' \emph{CoRR}, vol. abs/2407.06446, 2024. [Online].
  Available: \url{https://doi.org/10.48550/arXiv.2407.06446}
\BIBentrySTDinterwordspacing

\bibitem{FCC:conf/itw/PremlalR24}
\BIBentryALTinterwordspacing
R.~Premlal and B.~S. Rajan, ``On function-correcting codes,'' in \emph{{IEEE}
  Information Theory Workshop, {ITW} 2024, Shenzhen, China, November 24-28,
  2024}.\hskip 1em plus 0.5em minus 0.4em\relax {IEEE}, 2024, pp. 603--608.
  [Online]. Available: \url{https://doi.org/10.1109/ITW61385.2024.10807007}
\BIBentrySTDinterwordspacing

\bibitem{Coding:books/MacWilliamsS77}
M.~F. J. and S.~N.~J. A., \emph{The Theory of Error-Correcting Codes}.\hskip
  1em plus 0.5em minus 0.4em\relax Amsterdam: Elsevier, 1977.

\bibitem{DSS:journals/ftcit/RamkumarBSVKK22}
V.~Ramkumar, B.~S. Babu, B.~Sasidharan, M.~Vajha, M.~N. Krishnan, and P.~V.
  Kumar, ``Codes for distributed storage,'' \emph{Found. Trends Commun. Inf.
  Theory}, vol.~19, no.~4, pp. 547--813, 2022.

\bibitem{MDS:journals/ComLet/LacanF04}
J.~Lacan and J.~Fimes, ``Systematic mds erasure codes based on vandermonde
  matrices,'' \emph{IEEE Communications Letters}, vol.~8, no.~9, pp. 570--572,
  2004.

\end{thebibliography}


\begin{thebibliography}{99}


\bibitem{SRR:journals/tit/AktasJKKS21}
M. S. Aktas, G. Joshi, S. Kadhe, F. Kazemi, and E. Soljanin, 
"Service Rate Region: A New Aspect of Coded Distributed System Design," 
\textit{IEEE Transactions on Information Theory}, vol. 67, no. 12, pp. 7940--7963, 2021.

\bibitem{SRR:conf/isit/KazemiKS20}
F. Kazemi, S. Kurz, and E. Soljanin, 
"A Geometric View of the Service Rates of Codes Problem and its Application to the Service Rate of the First Order Reed-Muller Codes," 
in \textit{IEEE International Symposium on Information Theory (ISIT)}, Los Angeles, CA, USA, Jun. 2020, pp. 66--71.

\bibitem{SRR:conf/isit/KazemiKSS20}
F. Kazemi, E. Karimi, E. Soljanin, and A. Sprintson, 
"A Combinatorial View of the Service Rates of Codes Problem, its Equivalence to Fractional Matching and its Connection with Batch Codes," 
in \textit{IEEE International Symposium on Information Theory (ISIT)}, Los Angeles, CA, USA, Jun. 2020, pp. 646--651.

\bibitem{CodesForDSS:journals/ACMStorage/ZhirongYKPXYJ25}
Z. Shen, Y. Cai, K. Cheng, P. P. C. Lee, X. Li, Y. Hu, and J. Shu, 
"A Survey of the Past, Present, and Future of Erasure Coding for Storage Systems," 
\textit{ACM Transactions on Storage}, vol. 21, no. 1, Feb. 2025.

\bibitem{DSS:journals/pieee/DimakisRWS11}
A. G. Dimakis, K. Ramchandran, Y. Wu, and C. Suh, 
"A Survey on Network Codes for Distributed Storage," 
\textit{Proceedings of the IEEE}, vol. 99, no. 3, pp. 476--489, 2011.

\bibitem{Service:journals/siaga/AlfaranoKRS24}
G. N. Alfarano, A. B. Kilic, A. Ravagnani, and E. Soljanin, 
"The Service Rate Region Polytope," 
\textit{SIAM Journal on Applied Algebra and Geometry}, vol. 8, no. 3, pp. 553--582, 2024.

\bibitem{CodeParameters:conf/isit/KilicRS24}
A. B. Kilic, A. Ravagnani, and E. Soljanin, 
"On the Parameters of Codes for Data Access," 
in \textit{IEEE International Symposium on Information Theory (ISIT)}, Athens, Greece, Jul. 2024, pp. 819--824.

\bibitem{muller1954application}
D. E. Muller, 
"Application of Boolean Algebra to Switching Circuit Design and to Error Detection," 
\textit{Transactions of the IRE Professional Group on Electronic Computers}, no. 3, pp. 6--12, 1954.

\bibitem{reed1954class}
I. Reed, 
"A Class of Multiple-Error-Correcting Codes and the Decoding Scheme," 
\textit{Transactions of the IRE Professional Group on Information Theory}, vol. 4, no. 4, pp. 38--49, 1954.

\bibitem{fht}
Y. Be'ery and J. Snyders, 
"Optimal Soft Decision Block Decoders Based on Fast Hadamard Transform," 
\textit{IEEE Transactions on Information Theory}, vol. 32, no. 3, pp. 355--364, 1986.

\bibitem{fht2}
R. R. Green, 
"A Serial Orthogonal Decoder," 
\textit{JPL Space Programs Summary}, vol. 37--39-IV, pp. 247--253, 1966.

\bibitem{sidelnikov}
V. M. Sidel'nikov and A. S. Pershakov, 
"Decoding of Reed-Muller Codes with a Large Number of Errors," 
\textit{Problemy Peredachi Informatsii}, vol. 28, no. 3, pp. 80--94, 1992.

\bibitem{sakkour}
B. Sakkour, 
"Decoding of Second Order Reed-Muller Codes with a Large Number of Errors," 
in \textit{IEEE Information Theory Workshop}, 2005, p. 3.

\bibitem{dumer2}
I. Dumer, 
"Soft-Decision Decoding of Reed-Muller Codes: A Simplified Algorithm," 
\textit{IEEE Transactions on Information Theory}, vol. 52, no. 3, pp. 954--963, 2006.

\bibitem{ye2020recursive}
M. Ye and E. Abbe, 
"Recursive Projection-Aggregation Decoding of Reed-Muller Codes," 
\textit{IEEE Transactions on Information Theory}, vol. 66, no. 8, pp. 4948--4965, 2020.

\bibitem{kudekar2017reed}
S. Kudekar, S. Kumar, M. Mondelli, H. D. Pfister, E. Şaşoğlu, and R. Urbanke, 
"Reed-Muller Codes Achieve Capacity on Erasure Channels," 
\textit{IEEE Transactions on Information Theory}, vol. 63, no. 7, pp. 4298--4316, 2017.

\bibitem{Reeves}
G. Reeves and H. D. Pfister, 
"Reed-Muller Codes on BMS Channels Achieve Vanishing Bit-Error Probability for All Rates Below Capacity," 
\textit{IEEE Transactions on Information Theory}, pp. 920--949, 2023.

\bibitem{abbesandon}
E. Abbe and C. Sandon, 
"A Proof That Reed-Muller Codes Achieve Shannon Capacity on Symmetric Channels," 
in \textit{IEEE 64th Annual Symposium on Foundations of Computer Science (FOCS)}, 2023, pp. 177--193.

\bibitem{rm_survey}
E. Abbe, A. Shpilka, and M. Ye, 
"Reed-Muller Codes: Theory and Algorithms," 
\textit{IEEE Transactions on Information Theory}, vol. 67, no. 6, pp. 3251--3277, 2021.

\bibitem{5gnr}
"Final Report of 3GPP TSG RAN WG1 \#87 v1.0.0," 
available at: \url{http://www.3gpp.org/ftp/tsg ran/WG1 RL1/TSGR1 87/Report/}.

\bibitem{calderbank2010construction}
R. Calderbank, S. Howard, and S. Jafarpour, 
"Construction of a Large Class of Deterministic Sensing Matrices That Satisfy a Statistical Isometry Property," 
\textit{IEEE Journal of Selected Topics in Signal Processing}, vol. 4, no. 2, pp. 358--374, 2010.

\bibitem{beimel2005general}
A. Beimel, Y. Ishai, and E. Kushilevitz, 
"General Constructions for Information-Theoretic Private Information Retrieval," 
\textit{Journal of Computer and System Sciences}, vol. 71, no. 2, pp. 213--247, 2005.

\bibitem{Coding:books/MacWilliamsS77}
F. J. MacWilliams and N. J. A. Sloane, 
\textit{The Theory of Error-Correcting Codes}, Elsevier, Amsterdam, 1977.

\bibitem{NetworkCode:book/MarcusONA}
M. Greferath, M. O. Pavčević, N. Silberstein, and M. A. Vazquez-Castro, 
\textit{Network Coding and Subspace Designs}, Springer, 2018.
\end{thebibliography}
\bibliographystyle{IEEEtran}

\end{document}